\documentclass{article}

\usepackage[preprint]{neurips_2026}
\workshoptitle{Third Workshop on Agents in the Wild: Safety, Security, and Beyond}

\usepackage[utf8]{inputenc}
\usepackage[T1]{fontenc}
\usepackage{hyperref}
\hypersetup{hidelinks,hypertexnames=false}
\usepackage{url}
\usepackage{booktabs}
\usepackage{amsfonts}
\usepackage{amsmath}
\usepackage{amssymb}
\usepackage{amsthm}
\usepackage{microtype}
\usepackage{graphicx}
\usepackage{xcolor}
\usepackage{multirow}
\usepackage{array}
\usepackage{float}
\usepackage{placeins}

\theoremstyle{plain}
\newtheorem{proposition}{Proposition}
\theoremstyle{definition}

\newcommand{\DT}{D_T}
\newcommand{\DI}{D_I}
\newcommand{\Gap}{G}
\newcommand{\cmark}{\checkmark}
\newcommand{\xmark}{--}
\title{Structurally Close, Temporally Distant: Measuring Security Exposure in Long-Horizon LLM Agents}
\author{%
  Md Jafrin Hossain\thanks{Corresponding author: \texttt{mdjafrin.hossain@uwa.edu.au}} \\
  The University of Western Australia \\
  \And
  Nur Al Hasan Haldar \\
  The University of Western Australia \\
}

\begin{document}
\raggedbottom

\maketitle

\begin{abstract}
Long-horizon LLM agents interact with untrusted content, persistent memory, external state, and sensitive tools. Existing analyses commonly characterize an attack by the number of execution steps separating malicious input from a downstream action. We show that temporal remoteness can substantially overstate security separation in stateful agents. We introduce a provenance-aware execution graph that connects agent events through deterministic state, identifier, and tool provenance, and define \emph{influence distance} $\DI$ as the shortest structural path between an untrusted source and a sensitive action. We compare it against \emph{sequence distance} $\DT$, the shortest injection--sink path in the ordered trajectory. Because the influence graph contains every edge of the sequence graph, $\DI \leq \DT$; the gap $\Gap = \DT - \DI$ quantifies how much the step-count view overstates separation. Across 454 injection--sink pairs from 360 long-horizon AgentDojo trajectories on OpenAI's \texttt{gpt-4o-mini} and \texttt{gpt-4o} and Claude's Haiku 4.5 and Sonnet 4.6, $\Gap > 0$ for 96.9\% of pairs, with a median gap of 9 hops; 91.0\% remain decoupled after removing the largest provenance-only edge class. In a cross-domain evaluation on AgentDojo's banking suite, 33.8\% of 231 pairs from 377 trajectories decouple through different provenance mechanisms. Among the 274 pairs from \texttt{gpt-4o-mini} and \texttt{gpt-4o}, $\Gap$ does not independently predict attack success after controlling for $\DT$, attack family, and backend ($\beta_{\Gap}=0.066$, $p=.088$). At matched thresholds $k=2$ and $k=3$, a deterministic $\DI$-based pre-execution gate blocks five attack sinks missed by the sequence-only gate, with no additional benign blocking under either graph variant, though this paired gain is not significant at the 5\% level ($p=.0625$). Execution structure can therefore reveal proximity that step count hides, supporting more targeted runtime intervention. Hence, in this study, we measure candidate influence pathways rather than causal attribution.
\end{abstract}

\section{Introduction}
\label{sec:intro}

Large language model–based agents now operate across multiple stages, interleaving reasoning, tool calls, and interaction with external environments~\citep{yao2023react,schick2023toolformer,zhou2024webarena}. They consume documents, use tools, write to memory stores, recall what they wrote from memory stores, and eventually carry out meaningful actions, including sending emails, logging reports, and making payments. Much of their input data is taken from outside the trust boundary, and indirect prompt injection attacks exploit precisely this \citep{greshake2023not,debenedetti2024agentdojo}. For a 40-stage interaction history, the adversarial input at the front seems distant from the sensitive action at the back, and the notion that distance confers any security advantage implicitly surfaces whenever a long-horizon attack is framed in terms of the number of stages before execution.

This is because such intuition actually answers another question. Temporal distance does not equate to structural distance, as an instruction read twenty steps back could be written to a note at step 6 and accessed at step 24, just before the sensitive call was made. From the viewpoint of the sequence, there would be much distance between these two points, but structurally, there would be no distance at all. The number of intervening steps gives us a measure of how long ago something took place.

To quantify this formally, we consider two metrics on the same execution path. Sequence distance $\DT$ is the length of the shortest path from injection to sink in a graph whose edges represent consecutive executions. Influence distance $\DI$ is the shortest path length in the graph with additional deterministic provenance and state edges added from the log: tool output that another call refers to by identifier; state update and its subsequent return; and a batch of tool invocations sent in one assistant call. The structural gap between the two is $\Gap=\DT-\DI$. Both are calculated from execution paths without a language model to determine relevance, without an embedding similarity threshold, and without causal inference on any path.

\begin{figure}[tbp]
  \centering
  \includegraphics[width=0.8\linewidth]{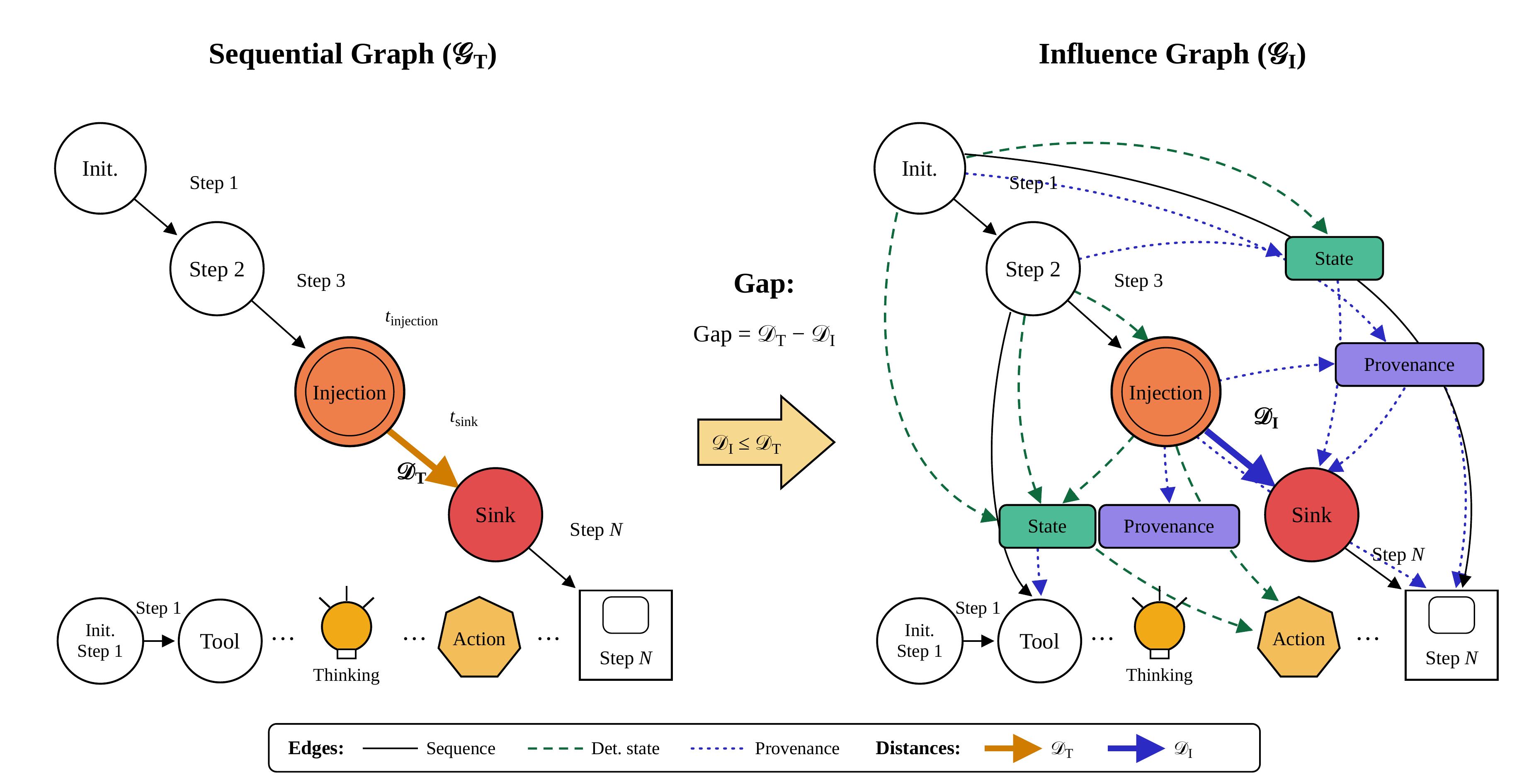}
  \caption{Method overview: constructing and comparing sequence and influence graphs.}
  \label{fig:overview}
\end{figure}

Figure~\ref{fig:overview} summarizes the complete pipeline, showing the sequential execution graph and the influence graph side by side.

\paragraph{Contributions.}
\textbf{(C1) Conceptual.} We distinguish temporal remoteness, structural exposure, and exploitability as separate properties of long-horizon agent security, and show that step count is a reliable proxy for neither of the other two.
\textbf{(C2) Method.} We introduce provenance-aware execution graphs and a deterministic influence distance $\DI$ that, paired with sequence distance $\DT$, yields a gap $\Gap = \DT - \DI$ computed from logs alone with no model in the loop.
\textbf{(C3) Empirical.} Across multiple agent environments and model families, we characterize when temporal and structural distance diverge and identify the provenance mechanisms responsible: decoupling holds for 96.9\% of long-horizon pairs, persists at 91.0\% after ablating the top contributor, yet drops to 33.8\% in banking where typed identifiers replace the memory pattern.
\textbf{(C4) Intervention.} Influence distance supports a deterministic pre-execution gate that picks up additional attack coverage over a sequence-only gate, and we characterize where that coverage starts costing benign actions.

\section{Background and Related Work}
\label{sec:related}

\paragraph{Agent injection benchmarks.}
Indirect prompt injection against tool-using agents is well documented \citep{greshake2023not,liu2024formalizing}, and several benchmarks measure it under controlled conditions. AgentDojo \citep{debenedetti2024agentdojo} provides task suites, injection tasks, and a utility--security protocol; InjecAgent \citep{zhan2024injecagent} targets tool-integrated agents; AgentHarm \citep{andriushchenko2025agentharm}, ToolEmu \citep{ruan2024toolemu}, and Agent-SafetyBench \citep{zhang2024agentsafetybench} evaluate harmful behavior more broadly. All report outcomes: whether the injection task succeeded, whether the user task was completed, and sometimes how long the interaction ran. Trajectories are recorded, but as a means to those outcomes. We take the trajectory itself as the object of study.

\paragraph{Interaction length as a proxy for separation.}
Closest to our question, long-horizon agent security work characterizes attacks by interaction length, the number of steps before a malicious instruction is executed, and by how success rates vary with it \citep{yi2025benchmarking,yao2024survey}. This describes attack dynamics but not proximity: interaction length records \emph{when} the sensitive action occurs relative to the injection, not how closely the two remain connected at the moment of execution. That gap is what we measure.

\paragraph{Information-flow and provenance defenses.}
Tracking untrusted information through a system is an old idea \citep{denning1976lattice,sabelfeld2003language,newsome2005dynamic} and an active one for agents. CaMeL \citep{debenedetti2025defeating} separates control and data flow by having a privileged component emit a program over quarantined values; FIDES \citep{costa2025securing} carries confidentiality and integrity labels through execution and refuses actions the label lattice forbids; RTBAS \citep{zhong2025rtbas} propagates taint across tool calls and surfaces only those decisions requiring user confirmation. Attention Tracker \citep{siu2025neurotaint} detects within the model rather than around it, and IFC has been argued for as a system-level framework for this threat \citep{wu2024system}. Others separate instructions from untrusted data at the model boundary \citep{chen2025struq} or isolate agent applications and execution contexts \citep{wu2025isolategpt}. Methodologically we are closest to whole-system provenance capture and inline analysis \citep{pasquier2017camflow,pasquier2018runtime} and to provenance-graph intrusion detection \citep{king2005provenance,han2020unicorn}: we adopt their stance of building structure from execution records rather than semantic inference.

\paragraph{Positioning.}
The previous defenses offer a answer to a policing question – is this behavior permissible given how it acquired its inputs. Ours is a measuring question -- how close does untrusted content really live to the action about to execute, and how much did the transcript mask that closeness. Since we only measure, our graph encodes that a path exists between the points, and says nothing regarding how content was leveraged. Table \ref{tab:positioning} lists the four dimensions that distinguish our setting from prior approaches.

\begin{table}[tb]
  \centering
    \caption{Positioning across long-horizon analysis, execution provenance, structural distance, and runtime intervention.}
  \label{tab:positioning}
  \small
  \begin{tabular}{lcccc}
    \toprule
    Work & Long horizon & Provenance & Struct.\ distance & Runtime \\
    \midrule
    AgentDojo \citep{debenedetti2024agentdojo}       & \cmark & \xmark & \xmark & \xmark \\
    Long-horizon character.\ \citep{yi2025benchmarking} & \cmark & \xmark & \xmark & \xmark \\
    CaMeL \citep{debenedetti2025defeating}           & \xmark & \cmark & \xmark & \cmark \\
    FIDES \citep{costa2025securing}                  & \xmark & \cmark & \xmark & \cmark \\
    RTBAS \citep{zhong2025rtbas}                     & \xmark & \cmark & \xmark & \cmark \\
    Attention Tracker \citep{siu2025neurotaint}      & \xmark & \xmark & \xmark & \cmark \\
    \midrule
    \textbf{This work}                               & \cmark & \cmark & \textbf{\cmark} & \textbf{\cmark} \\
    \bottomrule
  \end{tabular}
\end{table}

\section{Problem Formulation}
\label{sec:method}

\subsection{Setting and threat model}

An agent executes a user task over a long horizon. At each step it may read content, call a tool, write to or read from persistent state, or perform a \emph{sensitive action} such as sending a message, filing a record, or moving money. Execution is recorded as an ordered log of events $\tau = (v_1, \dots, v_n)$ over the event set $V$.

Two subsets of $V$ define the security problem. The \emph{untrusted sources} $S \subseteq V$ are events carrying content from outside the trust boundary: fetched documents, third-party tool results, and any injected instruction they contain. The \emph{sensitive sinks} $A \subseteq V$ are events invoking a tool designated security-relevant. Sink membership is fixed per suite in a static configuration of tool names, so nothing in the trajectory determines what counts as a sink.

The adversary controls the content of at least one $s \in S$ and aims to cause some $a \in A$ that the user did not request. It does not control the agent, the tools, or the runtime; it controls only text the agent will read. The defender observes the execution log and must decide, before $a$ executes, whether $a$ stands too close to some untrusted source to be permitted.

\subsection{Three notions of proximity}

This decision presupposes a notion of proximity between untrusted sources and sensitive actions, and the appropriate notion is not obvious. Three properties are routinely conflated under it, and separating them is contribution (C1).

The difficulty is in \emph{close}. Three properties are routinely conflated, and separating them is a contribution (C1).

\emph{Temporal remoteness} is how many steps separate $s$ from $a$ in execution order. It is cheap to compute and is the quantity long-horizon analyses report.

\emph{Structural exposure} is how few recorded derivation steps connect $s$ to $a$ through persistent state, identifiers, and execution dependencies. A note written at step 6 and recalled at step 24 is temporally remote but structurally adjacent, and step count cannot see the difference.

\emph{Exploitability} is whether an attack in fact succeeds. It turns on the model, the phrasing of the injection, and any defenses in place, all properties of the agent rather than of the trajectory's structure.

The defender's question is therefore not how long ago untrusted content arrived, but how few recorded steps stand between it and the action about to execute. The remainder of this section makes structural exposure precise: Section~\ref{sec:graph} constructs the graph it is measured on, Section~\ref{sec:distances} defines the measure and its relation to temporal remoteness, and Section~\ref{sec:gate} turns it into a pre-execution decision. Whether it also predicts exploitability is a separate empirical question, deferred to Section~\ref{sec:e2}.

\subsection{Execution graph}
\label{sec:graph}
Given a set of events $V = \{v_1, \dots, v_n\}$ we construct two graphs. The \emph{sequence graph} $G_T = (V, E_T)$ has a single edge type $E_T = \{ (v_i, v_{i+1}) : 1 \le i < n \}$, capturing only the fact that one event occurred after another. The \emph{influence graph} $G_I = (V, E_T \cup E_P)$ also has provenance edges $E_P$, created using four rules:

\begin{enumerate}
\itemsep1pt
\parskip0pt

  \item \textbf{State provenance.} A write to a state key and a subsequent read of that key are joined. Both the source and target sides of the write are logged by the runtime; hence, this is a lookup, not an inference.

  \item \textbf{Typed identifier reference.} A tool result returning an identifier (note ID, transaction ID, or contact record) and a subsequent tool call containing that identifier in its arguments are joined based on exact matching of the \texttt{source\_ids} field.

  \item \textbf{Batch provenance.} Tool calls emitted in the same assistant message are joined to the message because the runtime generated them from a single decision.

  \item \textbf{Same-tool call/result.} A tool call and its corresponding result are joined. This is the most numerous group, and we elaborate on it in Section~\ref{sec:e1}.

\end{enumerate}

\subsection{Sequence distance, influence distance, and the structural gap}
\label{sec:distances}
The two graphs induce two distances for any injection event $s$ and sensitive action $a$ in the same trajectory. \emph{Sequence distance} is $\DT(s,a) = d_{G_T}(s,a)$, the shortest-path length in the sequence graph, answering how far apart the injection and sink are if we look only at execution order.\ \emph{Influence distance} is $\DI(s,a) = d_{G_I}(s,a)$, the shortest-path length in the influence graph, answering how far apart they are once recorded provenance and state pathways are available too. In $G_T$, the only route from $s$ to $a$ runs along consecutive events, so $\DT(s,a)$ is the index difference $|i_a - i_s|$; the influence graph may admit shorter routes. Their difference is the \emph{structural gap} $\Gap(s,a) = \DT(s,a) - \DI(s,a)$, which is non-negative by the following.

\begin{proposition}[Monotonicity]
\label{prop:mono}
For every pair $(s,a)$ in every trajectory, $\DI(s,a) \le \DT(s,a)$.
\end{proposition}
\begin{proof}
$E_T \subseteq E_T \cup E_P = E_I$, so $G_T$ is a spanning subgraph of $G_I$. Every $s$--$a$ path in $G_T$ is an $s$--$a$ path in $G_I$ of the same length, so the minimum over $G_I$ paths is at most the minimum over $G_T$ paths.
\end{proof}

When $\Gap = 0$, provenance offers no route shorter than execution order, and the sequence view is structurally accurate for this pair. When $\Gap > 0$, a structural shortcut exists: some chain of recorded relationships connects injection to sink in fewer hops than the transcript suggests, and a large $\Gap$ means the sequence representation substantially overstates the separation. We deliberately do not read a large $\Gap$ as greater causal influence. The gap is a statement about the graph we constructed, and Section~\ref{sec:e2} tests and fails to find an association between gap magnitude and attack success.

\subsection{Guarding against circularity}
\label{sec:circularity}

A structural metric based on attack traces is vulnerable to a common criticism: if the path is traced based on a successful attack, then the metric itself includes the result that the metric intends to capture. We prevent this from happening in two ways. First, distances are computed for every eligible injection--sink pair in every injected trajectory regardless of attack outcome, so nothing about the outcome enters graph construction or path selection; the 454 pairs in Section~\ref{sec:e1} include both successful and unsuccessful attacks. The gate study separately labels decisions as attack or benign. Second, the existence of a path is not treated as evidence of causation anywhere. A short $\DI$ states that the log records a chain of relationships; it does not say the injected content shaped the model's decision.

\subsection{Runtime gate}
\label{sec:gate}

Both distances are cheap enough to evaluate immediately before a sensitive action executes. Define the influence gate at threshold $k$,
\[
\mathcal{G}_I(a; k) = \mathbf{1}\!\left[ \min_{s \in S} \DI(s,a) \le k \right],
\]
where $S$ is the set of untrusted sources seen so far in the trajectory, and $\mathcal{G}_I(a;k) = 1$ indicates that the action is being blocked. The baseline for sequences is exactly the same as for $\DI$ except using $\DT$ instead: $\mathcal{G}_T(a;k) = \mathbf{1}[\min_s \DT(s,a) \le k]$.

For the sets $\mathcal{B}_T(k)$ and $\mathcal{B}_I(k)$ of actions blocked by each gate, Proposition~\ref{prop:mono} implies $\DT(s,a) \le k \Rightarrow \DI(s,a) \le k$, and hence $\mathcal{B}_T(k) \subseteq \mathcal{B}_I(k)$ for all $k$. Set inclusion is thus a property of the construction, not an empirical result: the influence gate cannot miss an action the sequence gate blocks. What the construction leaves open is the size and contents of $\mathcal{B}_I(k) \setminus \mathcal{B}_T(k)$: how many are attack sinks, and how many are benign actions.

\section{Experimental Setup}
\label{sec:setup}

\paragraph{Primary evaluation.}
Our primary evaluation is based on the long-horizon task suite which is supported by the AgentDojo execution and logging framework \citep{debenedetti2024agentdojo}: there are 15 user tasks with 6 injection tasks and 4 model backends, resulting in 360 trajectories. Our agent uses \texttt{get\textbackslash document}, \texttt{save\_note}, \texttt{get\_note}, \texttt{search\_contacts}, \texttt{get\_calendar}, \texttt{send\_email}, \texttt{create\_report} methods, the last two of which are identified as sensitive. The first three injection tasks, 0-2, involve direct leakage as the sink appears right after the injection, while the remaining injection tasks, 3-5, use a delayed approach based on save and retrieve mechanism.
By considering only pairs where both distances are well-defined, 454 comparable pairs remain: 127 for \texttt{gpt-4o-mini}, 147 for \texttt{gpt-4o}, 91 for Claude Haiku 4.5, and 89 for Claude Sonnet 4.6. There are 90 trajectories per model. The trace-level success rates of the attacks are 50.0\%, 68.9\%, 1.1\%, and 0\%, respectively. Since the Anthropic backends have almost zero success rates, there is complete or near-complete separation in an ordinary pooled logistic regression; Section~\ref{sec:e2} thus presents the fully specified analysis with the two OpenAI models and only reports the four-model success rates descriptively.

\paragraph{Cross-domain validation.}
We additionally evaluate the standard AgentDojo banking suite as a within-framework cross-domain validation. It is not an independent benchmark: it shares AgentDojo's runtime and logging code with our primary evaluation, which limits what the comparison establishes. What varies are the task domain, the toolset, and, critically, the available provenance mechanisms. Across three available backends, it gives 377 traces and 231 comparable pairs: 99 for \texttt{gpt-4o-mini}, 54 for Haiku 4.5, and 78 for Sonnet 4.6. The evaluable injected-trace attack success rates are 74.1\% (40/54), 0\% (0/141), and 0.7\% (1/144), respectively. These labels incorporate a corrected evaluator polarity: AgentDojo banking's security predicate returns true on attack success, unlike the custom long-horizon predicate. Sensitive tools include \texttt{send\_money}, \texttt{schedule\_transaction}, \texttt{update\_scheduled\_transaction}, \texttt{update\_password}, and \texttt{update\_user\_info} sensitive. There is no save-and-retrieve memory pattern anywhere in this suite, which makes it a useful check.

\paragraph{Graph variants.}
Every structural result is reported under two constructions: the \emph{full} graph with all four edge classes from Section~\ref{sec:graph}, and the \emph{conservative} graph dropping same-tool call/result edges. We single out that class because it is the largest and the rule most open to the objection that it is permissive, a call and its result being arguably one event rather than two related ones. Removing it is our strongest robustness check. Unless stated otherwise, cross-suite comparisons use the conservative graph on both sides.

\paragraph{Statistics.}
Proportions carry Wilson 95\% confidence intervals \citep{wilson1927probable}, which behave sensibly at the extreme rates several cells exhibit. Binned attack-success plots omit cells with $n < 10$ rather than showing uninformatively wide intervals; the regression uses all 274 OpenAI pairs. Because backends differ in both attack success and trajectory length, we analyse per backend, where the quantity is backend-dependent, and include backend as a covariate when we pool. Paired gate comparisons use an exact McNemar test \citep{mcnemar1947note}.

\section{E1: Sequence and Influence Distance Decouple}
\label{sec:e1}

\paragraph{RQ1.} Do sequence distance and influence distance diverge in agent
trajectories, and if so, by how much?

\subsection{Main result}

For the 454 comparable injection-sink pairs in the full graph, the two distances differ in 440 cases, yielding a ratio of 96.9\% (95\% confidence interval [94.9\%, 98.2\%]). The median of the difference is 9 hops, average 9.90 (standard deviation 6.21), with the maximum of 37 hops, while the median $\DT$ is 23 and the median $\DI$ is 12 (Pearson $r = 0.779$, Spearman $\rho = 0.753$).

Figure~\ref{fig:decoupling}(a) shows the pairs. In effect, all the mass lies below the diagonal, with substantial vertical variation at constant $\DT$: for example, two trajectories with an injection and a sink separated by 30 steps could have influence distances of 8 and 25.

\begin{figure}[tb]
  \centering
  \includegraphics[width=0.8\linewidth]{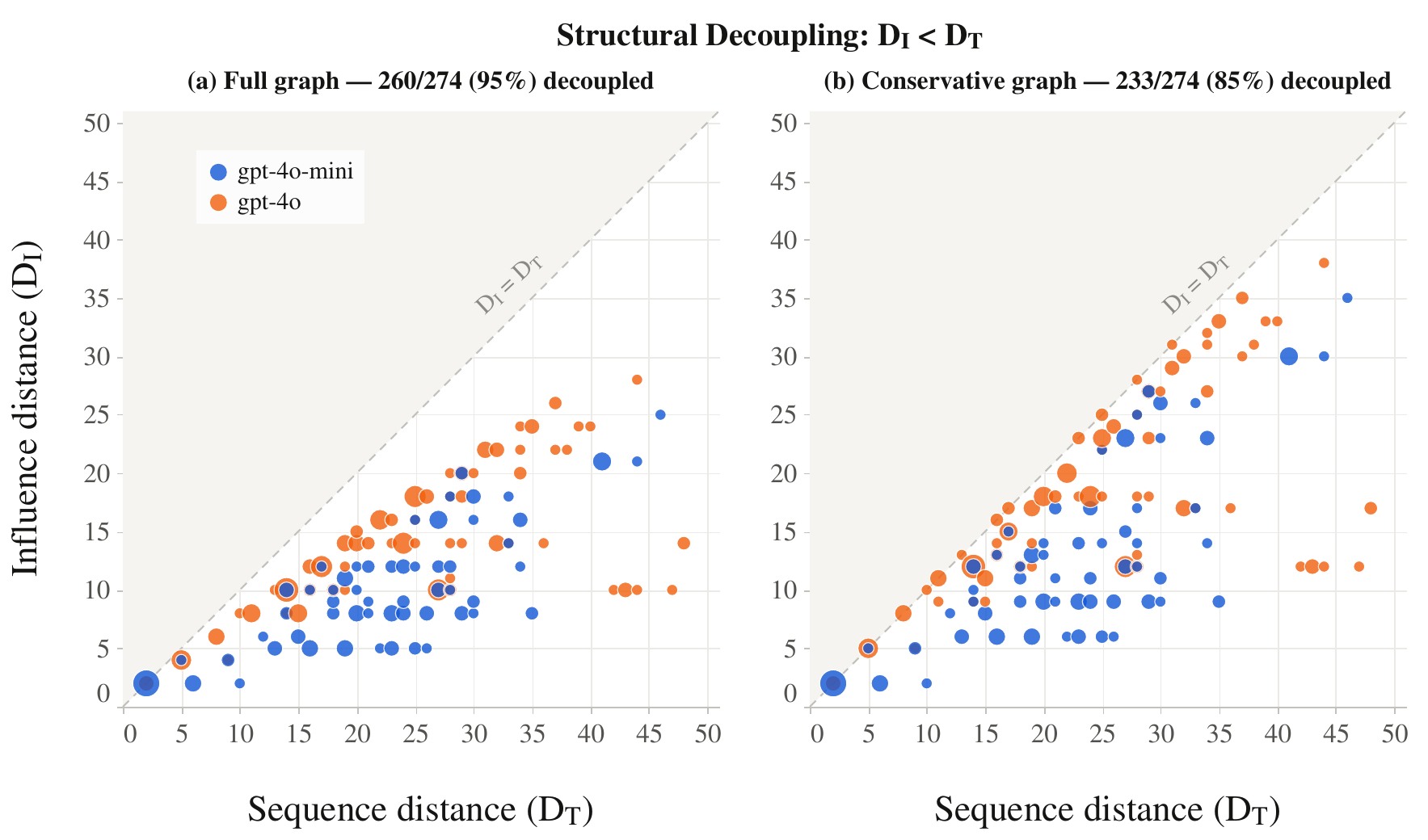}
  \caption{Sequence against influence distance for the 274 OpenAI pairs. The dashed line is $\DI = \DT$; Proposition~\ref{prop:mono} forbids points above it. Across all four models, 440/454 (96.9\%) decouple in the full graph and 413/454 (91.0\%) in the conservative graph.}
  \label{fig:decoupling}
\end{figure}

\subsection{Robustness to graph construction}

To test whether the result depends on same-tool call/result edges, we remove this edge class from the influence graph. The ablated graph leaves 413 of 454 pairs decoupled (91.0\%, 95\% CI [88.0\%, 93.3\%]), with a median gap of 3, a mean of 5.96 (SD 6.17), and a maximum of 35. In Figure~\ref{fig:decoupling}(b), the distribution moves toward the diagonal, as expected, but most pairs remain below it. Same-tool edges therefore increase the magnitude of decoupling without being necessary for it: removing them reduces the decoupling rate by 5.9 percentage points and the median gap by 6 hops. The remaining decoupling arises from state writes and reads, typed identifier references, and batch provenance.

\paragraph{Provenance verification.}
Since the measurement rests on edges being real rather than artifacts of our parsing, we sampled 15 shortcut-containing trajectories at random and checked each shortcut edge on the $s$--$a$ path against the canonical \texttt{source\_ids} recorded by the runtime; every edge corresponded to a relationship the log had recorded independently of our construction. The audit also clarified a scope condition: decoupling requires a sufficiently long sequence for a shortcut to beat the transcript route, so in short trajectories, $\Gap$ is necessarily small. 

\subsection{Cross-domain validation}

\begin{table}[tb]
  \centering
  \caption{Cross-domain comparison, conservative graph on both sides. Banking decouples at a lower rate through different provenance mechanisms, and its trajectories are much shorter.}
  \label{tab:crossdomain}
  \small
  \begin{tabular}{lrrrrrr}
    \toprule
    Suite & Pairs & Decoupled & Rate & Median $\Gap$ & Max $\Gap$ & Median $\DT$ \\
    \midrule
    Long-horizon & 454 & 413 & 91.0\% & 3 & 35 & 23 \\
    Banking      & 231 &  78 & 33.8\% & 0 & 19 &  5 \\
    \bottomrule
  \end{tabular}
\end{table}

Table~\ref{tab:crossdomain} compares the suites. Banking decouples in 78 of 231 pairs (33.8\%) with a median gap of 0 and a maximum of 19, far below the long-horizon rate; the reason is in the last column, where banking's median $\DT$ of 5 against 23 leaves much less room for a shortcut to exist.

What makes the comparison informative is the mechanism, not the rate. The long-horizon shortcuts primarily use note save-and-retrieve operations, along with contact and calendar state. Banking has no comparable memory pattern: its shortcuts use batch provenance, typed scheduled-transaction identifiers, and account-state mutations. This supports a narrow conclusion: decoupling is not confined to the long-horizon task design or to explicit save-and-retrieve pathways. Because both suites share AgentDojo's logging infrastructure, the evaluation does not establish generalization across independently developed agent benchmarks.

\section{E2: Does the Structural Gap Predict Attack Success?}
\label{sec:e2}

\paragraph{RQ2.} At a fixed sequence distance, is a larger structural gap
associated with a higher probability of attack success?

If a large gap indicates that untrusted content sits structurally close to a sink despite appearing distant, one might expect attacks with large gaps to succeed more often. We tested this and did not find it.

\begin{figure}[tb]
  \centering
  \includegraphics[width=0.65\linewidth]{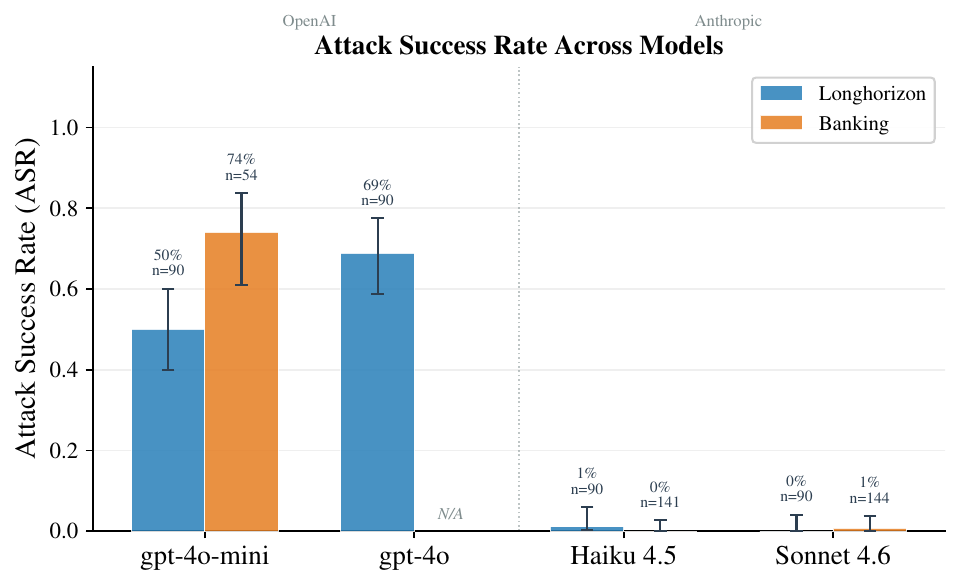}
  \caption{Corrected trace-level attack success rates with Wilson 95\% intervals. OpenAI models are susceptible across evaluated domains, while both Anthropic models are at 0--1.1\%. GPT-4o was not run on banking.}
  \label{fig:asrgap}
\end{figure}

Figure~\ref{fig:asrgap} gives the corrected four-model context. Long-horizon ASR is 50.0\%, 68.9\%, 1.1\%, and 0\% for GPT-4o-mini, GPT-4o, Haiku 4.5, and Sonnet 4.6. Banking ASR is 74.1\%, 0\%, and 0.7\% for the three evaluated models. Thus, model identity dominates the descriptive outcome differences and must not be confused with a structural-gap effect.

To separate the gap's contribution from the covariates it travels with, we fit a logistic model over the 274 OpenAI-model pairs:
\[
\operatorname{logit} \Pr(Y = 1) = \beta_0 + \beta_1 \DT + \beta_2 \Gap + \beta_3 \,\mathrm{Delayed} + \beta_4 \,\mathrm{Backend},
\]
where $Y$ indicates attack success, $\mathrm{Delayed}$ is a binary attack-family indicator separating the delayed injection tasks (3--5) from the immediate ones (0--2), and $\mathrm{Backend}$ indicates backend B. Results are in Table~\ref{tab:regression}.

\begin{table}[tb]
  \centering
  \caption{Logistic regression of attack success on structural gap with controls ($n = 274$, pseudo-$R^2 = 0.14$, AIC $= 283.0$). The gap coefficient is not significant at the 5\% level.}
  \label{tab:regression}
  \small
  \begin{tabular}{lrrrr}
    \toprule
    Predictor & Coef. & SE & $z$ & $p$ \\
    \midrule
    Intercept                & $1.645$  & $0.415$ & $3.97$  & $<.001$ \\
    Sequence distance $\DT$  & $-0.110$ & $0.029$ & $-3.84$ & $<.001$ \\
    Structural gap $\Gap$    & $0.066$  & $0.039$ & $1.71$  & $.088$ \\
    Delayed attack family    & $1.168$  & $0.311$ & $3.75$  & $<.001$ \\
    Backend B                & $1.334$  & $0.341$ & $3.92$  & $<.001$ \\
    \bottomrule
  \end{tabular}
\end{table}

All three control covariates are significant: longer sequence distance reduces attack success ($\beta = -0.110$, $p < 0.001$), delayed attacks succeed more often than immediate ones ($\beta = 1.168$, $p < 0.001$), and backend B is more susceptible than backend A ($\beta = 1.334$, $p < 0.001$). The structural gap is not ($\beta_\Gap = 0.066$, $p = 0.088$). We therefore do not find statistically significant evidence that the structural gap independently predicts attack success after these controls; with $n = 274$ and pseudo-$R^2$ of 0.14, the data do not resolve whether this coefficient differs from zero.

A structural shortcut can be useful for identifying candidate influence pathways at an action boundary, but it does not function as an independent vulnerability score. Characterization and intervention are different jobs, and evidence for one is not evidence for the other. Section~\ref{sec:e3} therefore does not ask whether $\Gap$ ranks actions by risk; it asks the narrower operational question of what a $\DI$ threshold catches that a $\DT$ threshold does not.

\section{E3: A Runtime Influence-Distance Gate}
\label{sec:e3}

\paragraph{RQ3.} What additional security coverage does influence distance
provide over a sequence-only gate, and at what cost in benign actions blocked?

\paragraph{Gate and evaluation protocol.}
The gate is the one from Section~\ref{sec:gate}: immediately before a sensitive action, compute the minimum influence distance from any untrusted source seen so far and block if it is at most $k$. Because the gate uses only recorded structure, it involves no model call, learned component, or randomness. Evaluation is a counterfactual replay with one decision per injection--sink pair, covering the same 274 OpenAI-model pairs as Section~\ref{sec:e2} (201 attack, 73 benign), comparing $\mathcal{G}_I$ against $\mathcal{G}_T$ at matched $k$. Since $\mathcal{B}_T(k) \subseteq \mathcal{B}_I(k)$ holds analytically, the influence gate cannot lose on any single decision and reporting a win rate would be reporting Proposition~\ref{prop:mono}; the informative quantities are how many extra attack sinks it catches and how many extra benign actions it blocks.

\subsection{Coverage against cost}

\begin{figure}[tb]
  \centering
  \includegraphics[width=0.65\linewidth]{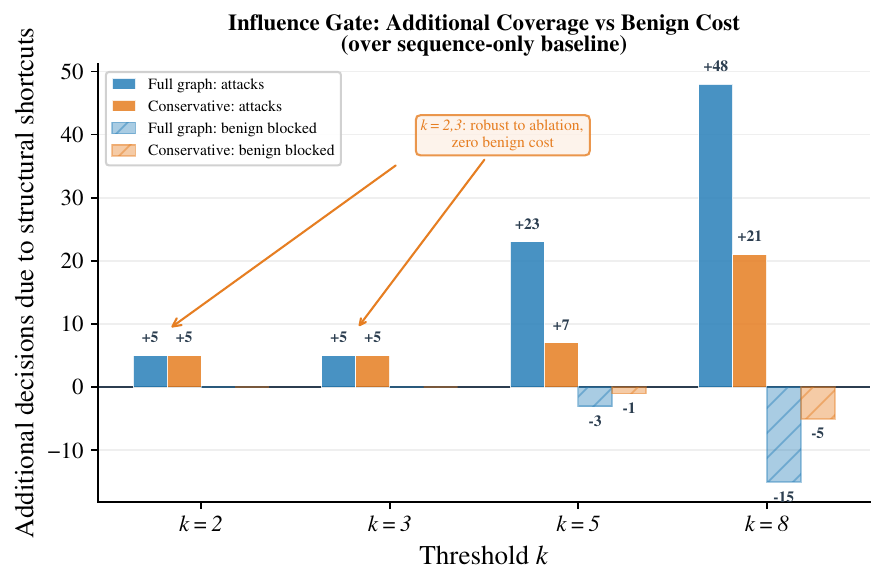}
\caption{Additional attack and benign decisions blocked by the influence gate at matched thresholds under \textbf{(a)} full and \textbf{(b)} conservative graphs. Solid and hatched bars denote attack and benign decisions, respectively.}
  \label{fig:coverage}
\end{figure}

Figure~\ref{fig:coverage} and Table~\ref{tab:gate} give the trade-off. At $k = 2$, the influence gate flags five attack sinks that the sequence gate misses and blocks no additional benign actions, and the same holds at $k = 3$. The conservative graph gives identical numbers at both thresholds, so this operating point does not depend on the edge rule we considered least reliable. 

Beyond that, the trade-off degrades: at $k = 5$, the full graph buys 23 extra attack sinks for 3 extra benign blocks compared with the conservative graph's 7 for 1; at $k = 8$, the corresponding trade-offs are 48 for 15 and 21 for 5.

\begin{table}[tb]
  \centering
  \caption{Additional decisions blocked by the influence gate relative to the sequence-only gate at matched $k$. The $k=3$ row equals $k=2$.}
  \label{tab:gate}
  \small
  \begin{tabular}{lrrrr}
    \toprule
    & \multicolumn{2}{c}{Full graph} & \multicolumn{2}{c}{Conservative} \\
    \cmidrule(lr){2-3}\cmidrule(lr){4-5}
    $k$ & Attack & Benign & Attack & Benign \\
    \midrule
    2 & +5 & +0 & +5 & +0 \\
    5 & +23 & +3 & +7 & +1 \\
    8 & +48 & +15 & +21 & +5 \\
    \bottomrule
  \end{tabular}
\end{table}

At $k=2$, exact McNemar testing with five discordant attack decisions and none in the reverse direction gives a two-sided $p=.0625$; the gain is operationally interesting but not statistically significant at 5\%. The gate is not a standalone prompt-injection defense. It is a model-free pre-execution check that surfaces candidate pathways that transcript position alone does not reveal.

\section{Discussion and Limitations}
\label{sec:discussion}
The central distinction of the study lies in the difference between the two notions of distance, not merely in their numerical values. Sequence distance asks how many execution steps ago the injection occurred, whereas influence distance asks how far away the injection is in structural terms. Both distances are legitimate measures, but they differ by a median of 9 hops in our dataset. Moreover, the observed difference should not be interpreted as a universal vulnerability score. It does not independently predict attack success, and the gate result is limited to intervention at the time of action. The distances are determined by connections between nodes in the graph, while the existence of a recorded path does not necessarily indicate that the agent accessed or used the injected data.

\section{Conclusion}

We introduced influence distance to measure how structurally close untrusted content remains to sensitive actions in LLM agent trajectories. Across two task domains, execution structure often reveals substantially shorter paths than sequence position alone, and influence-based gating catches attacks that sequence-only gating misses. These findings show that temporal distance should not be treated as structural separation.

\FloatBarrier

\bibliographystyle{plainnat}
{\footnotesize
\setlength{\bibsep}{1pt plus 0.3ex}
\bibliography{references}

@inproceedings{debenedetti2024agentdojo,
  title     = {{AgentDojo}: A Dynamic Environment to Evaluate Prompt Injection Attacks and Defenses for {LLM} Agents},
  author    = {Debenedetti, Edoardo and Zhang, Jie and Balunovi{\'c}, Mislav and Beurer-Kellner, Luca and Fischer, Marc and Tram{\`e}r, Florian},
  booktitle = {Advances in Neural Information Processing Systems Datasets and Benchmarks Track},
  year      = {2024}
}

@inproceedings{debenedetti2025defeating,
  title     = {Defeating Prompt Injections by Design},
  author    = {Debenedetti, Edoardo and Shumailov, Ilia and Fan, Tianqi and Hayes, Jamie and Carlini, Nicholas and Fabian, Daniel and Kern, Christoph and Shi, Chongyang and Terzis, Andreas and Tram{\`e}r, Florian},
  booktitle = {IEEE Conference on Secure and Trustworthy Machine Learning (SaTML)},
  year      = {2026},
  eprint    = {2503.18813},
  archivePrefix = {arXiv},
  primaryClass  = {cs.CR}
}

@article{costa2025securing,
  title   = {Securing {AI} Agents with Information-Flow Control},
  author  = {Costa, Manuel and K{\"o}pf, Boris and Kolluri, Aashish and Paverd, Andrew and Russinovich, Mark and Salem, Ahmed and Tople, Shruti and Wutschitz, Lukas and Zanella-B{\'e}guelin, Santiago},
  journal = {arXiv preprint arXiv:2505.23643},
  year    = {2025}
}

@article{zhong2025rtbas,
  title   = {{RTBAS}: Defending {LLM} Agents Against Prompt Injection and Privacy Leakage},
  author  = {Zhong, Peter Yong and Chen, Siyuan and Wang, Ruiqi and McCall, McKenna and Titzer, Ben L. and Miller, Heather and Gibbons, Phillip B.},
  journal = {arXiv preprint arXiv:2502.08966},
  year    = {2025}
}

@inproceedings{siu2025neurotaint,
  title     = {Attention Tracker: Detecting Prompt Injection Attacks in {LLM}s},
  author    = {Hung, Kuo-Han and Ko, Ching-Yun and Rawat, Ambrish and Chung, I-Hsin and Hsu, Winston H. and Chen, Pin-Yu},
  booktitle = {Findings of the Association for Computational Linguistics: NAACL 2025},
  pages     = {2309--2322},
  publisher = {Association for Computational Linguistics},
  doi       = {10.18653/v1/2025.findings-naacl.123},
  year      = {2025}
}

@inproceedings{yao2023react,
  title     = {{ReAct}: Synergizing Reasoning and Acting in Language Models},
  author    = {Yao, Shunyu and Zhao, Jeffrey and Yu, Dian and Du, Nan and Shafran, Izhak and Narasimhan, Karthik and Cao, Yuan},
  booktitle = {International Conference on Learning Representations},
  year      = {2023}
}

@inproceedings{schick2023toolformer,
  title     = {{Toolformer}: Language Models Can Teach Themselves to Use Tools},
  author    = {Schick, Timo and Dwivedi-Yu, Jane and Dess{\`i}, Roberto and Raileanu, Roberta and Lomeli, Maria and Hambro, Eric and Zettlemoyer, Luke and Cancedda, Nicola and Scialom, Thomas},
  booktitle = {Advances in Neural Information Processing Systems},
  volume    = {36},
  pages     = {68539--68551},
  year      = {2023}
}

@inproceedings{zhou2024webarena,
  title     = {{WebArena}: A Realistic Web Environment for Building Autonomous Agents},
  author    = {Zhou, Shuyan and Xu, Frank F. and Zhu, Hao and Zhou, Xuhui and Lo, Robert and Sridhar, Abishek and Cheng, Xianyi and Ou, Tianyue and Bisk, Yonatan and Fried, Daniel and Alon, Uri and Neubig, Graham},
  booktitle = {International Conference on Learning Representations},
  year      = {2024}
}

@inproceedings{chen2025struq,
  title     = {{StruQ}: Defending Against Prompt Injection with Structured Queries},
  author    = {Chen, Sizhe and Piet, Julien and Sitawarin, Chawin and Wagner, David},
  booktitle = {34th USENIX Security Symposium (USENIX Security 25)},
  pages     = {2383--2400},
  publisher = {USENIX Association},
  year      = {2025}
}

@inproceedings{wu2025isolategpt,
  title     = {{IsolateGPT}: An Execution Isolation Architecture for {LLM}-Based Agentic Systems},
  author    = {Wu, Yuhao and Roesner, Franziska and Kohno, Tadayoshi and Zhang, Ning and Iqbal, Umar},
  booktitle = {Network and Distributed System Security Symposium},
  year      = {2025}
}

@inproceedings{pasquier2017camflow,
  title     = {Practical Whole-System Provenance Capture},
  author    = {Pasquier, Thomas and Han, Xueyuan and Goldstein, Mark and Moyer, Thomas and Eyers, David and Seltzer, Margo and Bacon, Jean},
  booktitle = {Proceedings of the 2017 Symposium on Cloud Computing},
  pages     = {405--418},
  publisher = {Association for Computing Machinery},
  doi       = {10.1145/3127479.3129249},
  year      = {2017}
}

@inproceedings{pasquier2018runtime,
  title     = {Runtime Analysis of Whole-System Provenance},
  author    = {Pasquier, Thomas and Han, Xueyuan and Moyer, Thomas and Bates, Adam and Hermant, Olivier and Eyers, David and Bacon, Jean and Seltzer, Margo},
  booktitle = {Proceedings of the 2018 ACM SIGSAC Conference on Computer and Communications Security},
  pages     = {1601--1616},
  publisher = {Association for Computing Machinery},
  doi       = {10.1145/3243734.3243776},
  year      = {2018}
}

@inproceedings{greshake2023not,
  title     = {Not What You've Signed Up For: Compromising Real-World {LLM}-Integrated Applications with Indirect Prompt Injection},
  author    = {Greshake, Kai and Abdelnabi, Sahar and Mishra, Shailesh and Endres, Christoph and Holz, Thorsten and Fritz, Mario},
  booktitle = {Proceedings of the 2023 ACM SIGSAC Conference on Computer and Communications Security},
  doi       = {10.1145/3605764.3623985},
  year      = {2023}
}

@inproceedings{zhan2024injecagent,
  title     = {{InjecAgent}: Benchmarking Indirect Prompt Injections in Tool-Integrated Large Language Model Agents},
  author    = {Zhan, Qiusi and Liang, Zhixiang and Ying, Zifan and Kang, Daniel},
  booktitle = {Findings of the Association for Computational Linguistics: ACL 2024},
  pages     = {10471--10506},
  publisher = {Association for Computational Linguistics},
  doi       = {10.18653/v1/2024.findings-acl.624},
  year      = {2024}
}

@inproceedings{andriushchenko2025agentharm,
  title     = {{AgentHarm}: A Benchmark for Measuring Harmfulness of {LLM} Agents},
  author    = {Andriushchenko, Maksym and Souly, Alexandra and Dziemian, Mateusz and Duenas, Derek and Lin, Maxwell and Wang, Justin and Hendrycks, Dan and Zou, Andy and Kolter, Zico and Fredrikson, Matt and Winsor, Eric and Wynne, Jerome and Gal, Yarin and Davies, Xander},
  booktitle = {International Conference on Learning Representations},
  year      = {2025}
}

@inproceedings{liu2024formalizing,
  title     = {Formalizing and Benchmarking Prompt Injection Attacks and Defenses},
  author    = {Liu, Yupei and Jia, Yuqi and Geng, Runpeng and Jia, Jinyuan and Gong, Neil Zhenqiang},
  booktitle = {33rd USENIX Security Symposium},
  pages     = {1831--1847},
  publisher = {USENIX Association},
  year      = {2024}
}

@article{wu2024system,
  title   = {System-Level Defense against Indirect Prompt Injection Attacks: An Information Flow Control Perspective},
  author  = {Wu, Fangzhou and Cecchetti, Ethan and Xiao, Chaowei},
  journal = {arXiv preprint arXiv:2409.19091},
  year    = {2024}
}

@inproceedings{yi2025benchmarking,
  title     = {Benchmarking and Defending against Indirect Prompt Injection Attacks on Large Language Models},
  author    = {Yi, Jingwei and Xie, Yueqi and Zhu, Bin and Kiciman, Emre and Sun, Guangzhong and Xie, Xing and Wu, Fangzhao},
  booktitle = {Proceedings of the 31st ACM SIGKDD Conference on Knowledge Discovery and Data Mining},
  doi       = {10.1145/3690624.3709179},
  year      = {2025}
}

@article{denning1976lattice,
  title   = {A Lattice Model of Secure Information Flow},
  author  = {Denning, Dorothy E.},
  journal = {Communications of the ACM},
  volume  = {19},
  number  = {5},
  pages   = {236--243},
  doi     = {10.1145/360051.360056},
  year    = {1976}
}

@article{sabelfeld2003language,
  title   = {Language-Based Information-Flow Security},
  author  = {Sabelfeld, Andrei and Myers, Andrew C.},
  journal = {IEEE Journal on Selected Areas in Communications},
  volume  = {21},
  number  = {1},
  pages   = {5--19},
  doi     = {10.1109/JSAC.2002.806121},
  year    = {2003}
}

@inproceedings{newsome2005dynamic,
  title     = {Dynamic Taint Analysis for Automatic Detection, Analysis, and Signature Generation of Exploits on Commodity Software},
  author    = {Newsome, James and Song, Dawn},
  booktitle = {Proceedings of the 12th Annual Network and Distributed System Security Symposium},
  year      = {2005}
}

@inproceedings{king2005provenance,
  title     = {Backtracking Intrusions},
  author    = {King, Samuel T. and Chen, Peter M.},
  booktitle = {Proceedings of the 19th ACM Symposium on Operating Systems Principles},
  pages     = {223--236},
  doi       = {10.1145/945445.945467},
  year      = {2003}
}

@inproceedings{han2020unicorn,
  title     = {{UNICORN}: Runtime Provenance-Based Detector for Advanced Persistent Threats},
  author    = {Han, Xueyuan and Pasquier, Thomas and Bates, Adam and Mickens, James and Seltzer, Margo},
  booktitle = {Network and Distributed System Security Symposium},
  doi       = {10.14722/ndss.2020.24046},
  year      = {2020}
}

@article{mcnemar1947note,
  title   = {Note on the Sampling Error of the Difference between Correlated Proportions or Percentages},
  author  = {McNemar, Quinn},
  journal = {Psychometrika},
  volume  = {12},
  number  = {2},
  pages   = {153--157},
  doi     = {10.1007/BF02295996},
  year    = {1947}
}

@article{wilson1927probable,
  title   = {Probable Inference, the Law of Succession, and Statistical Inference},
  author  = {Wilson, Edwin B.},
  journal = {Journal of the American Statistical Association},
  volume  = {22},
  number  = {158},
  pages   = {209--212},
  doi     = {10.1080/01621459.1927.10502953},
  year    = {1927}
}

@article{yao2024survey,
  title   = {A Survey on Large Language Model ({LLM}) Security and Privacy: The Good, the Bad, and the Ugly},
  author  = {Yao, Yifan and Duan, Jinhao and Xu, Kaidi and Cai, Yuanfang and Sun, Zhibo and Zhang, Yue},
  journal = {High-Confidence Computing},
  volume  = {4},
  number  = {2},
  pages   = {100211},
  doi     = {10.1016/j.hcc.2024.100211},
  year    = {2024}
}

@article{zhang2024agentsafetybench,
  title   = {{Agent-SafetyBench}: Evaluating the Safety of {LLM} Agents},
  author  = {Zhang, Zhexin and Cui, Shiyao and Lu, Yida and Zhou, Jingzhuo and Yang, Junxiao and Wang, Hongning and Huang, Minlie},
  journal = {arXiv preprint arXiv:2412.14470},
  year    = {2024}
}

@inproceedings{ruan2024toolemu,
  title     = {Identifying the Risks of {LM} Agents with an {LM}-Emulated Sandbox},
  author    = {Ruan, Yangjun and Dong, Honghua and Wang, Andrew and Pitis, Silviu and Zhou, Yongchao and Ba, Jimmy and Dubois, Yann and Maddison, Chris J. and Hashimoto, Tatsunori},
  booktitle = {International Conference on Learning Representations},
  year      = {2024}
}
}

\clearpage
\appendix
\section*{Appendix}

\section{Worked Trajectory Example}
This appendix walks through one trajectory in full. In the 26-event run of Figure~\ref{fig:example}, the injection appears 16 steps before the sensitive \texttt{send\_email} call, so the transcript suggests substantial separation; the influence graph records a five-hop pathway between them ($\Delta T = 16$, $\Delta I = 5$, $\text{Gap} = 11$).

\begin{figure}[H]
  \centering
  \includegraphics[width=0.82\linewidth]{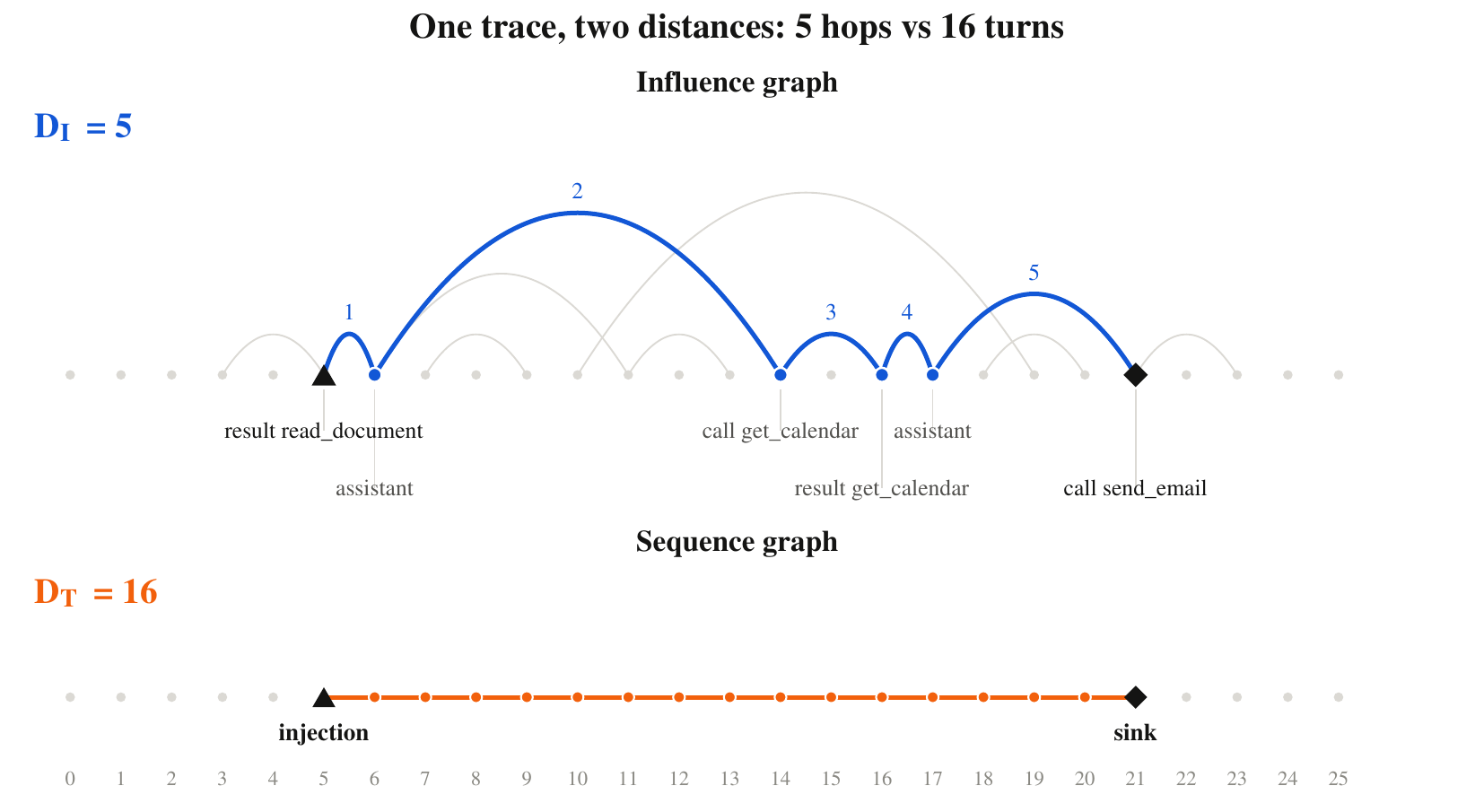}
  \caption{A 26-event long-horizon trajectory. The sensitive
  \texttt{send\_email} call has $\DT=16$, $\DI=5$, and $\Gap=11$.}
  \label{fig:example}
\end{figure}

\FloatBarrier

\section{Additional Model and Domain Results}

The aggregate result masks some variation across models. Table~\ref{tab:appendix-models} reports decoupling results by model. Long-horizon decoupling remains high across all models on the full graph, while the shorter banking trajectories exhibit more model-dependent graph structure and consequently lower conservative rates.

\begin{table}
  \centering
  \caption{Decoupling by model. LH uses the full graph; banking uses the conservative graph. \emph{(A dash indicates that the model was not evaluated.)}}
  \label{tab:appendix-models}
  \small
  \begin{tabular}{lrrrr}
    \toprule
    Model & LH pairs & LH rate & Bank pairs & Bank rate \\
    \midrule
    GPT-4o-mini & 127 & 91.3\% & 99 & 45.5\% \\
    GPT-4o      & 147 & 98.0\% & -- & -- \\
    Haiku 4.5   & 91  & 100.0\% & 54 & 18.5\% \\
    Sonnet 4.6  & 89  & 100.0\% & 78 & 29.5\% \\
    \bottomrule
  \end{tabular}
\end{table}

\begin{figure}[tbp]
  \centering
  \includegraphics[width=0.76\linewidth]{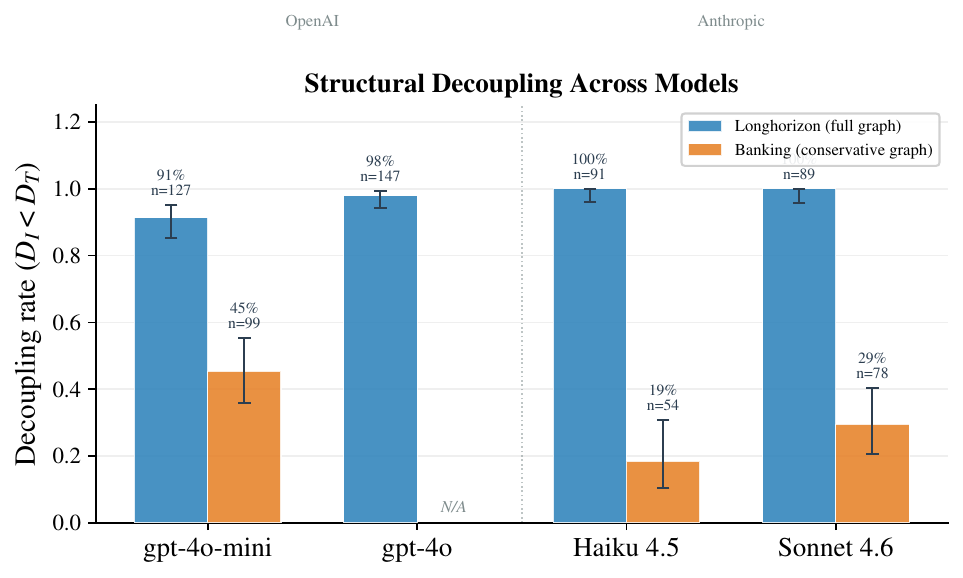}
  \caption{Structural decoupling across the four evaluated models. Error bars are Wilson 95\% confidence intervals.}
  \label{fig:appendix-multimodel}
\end{figure}

\section{Evaluator Correction and Reproducibility Notes}

The banking evaluation required an enhancement of the label polarity. The custom long-horizon security predicate returns \texttt{true} when the trajectory is secure, so attack success is its negation. AgentDojo's banking predicate uses the opposite convention, returning \texttt{true} when an attack succeeds. Thus, applying the long-horizon conversion inverted the banking labels. We modified this before generating the reported results and verified it manually. For manual verification, a refusal without a \texttt{send\_money} call was marked successful, while a transfer to the attacker was marked unsuccessful.

Table~\ref{tab:appendix-asr} gives the final trace-level rates. This substantially changes the interpretation: the Anthropic models are robust in both domains, whereas the evaluated OpenAI models exhibit much higher attack success.

\begin{table}[tbp]
  \centering
  \caption{Corrected attack success rates from evaluable injected traces.}
  \label{tab:appendix-asr}
  \small
  \begin{tabular}{lrrr}
    \toprule
    Suite & Model & Successful / traces & ASR \\
    \midrule
    Long-horizon & GPT-4o-mini & 45/90 & 50.0\% \\
    Long-horizon & GPT-4o      & 62/90 & 68.9\% \\
    Long-horizon & Haiku 4.5   & 1/90  & 1.1\% \\
    Long-horizon & Sonnet 4.6  & 0/90  & 0.0\% \\
    Banking & GPT-4o-mini & 40/54  & 74.1\% \\
    Banking & Haiku 4.5   & 0/141  & 0.0\% \\
    Banking & Sonnet 4.6  & 1/144  & 0.7\% \\
    \bottomrule
  \end{tabular}
\end{table}

\paragraph{Scope notes.}
Both domains use AgentDojo's execution and logging infrastructure, so the cross-domain result does not establish generalization to independently built agent runtimes. The evaluated threat model is indirect prompt injection through tool outputs. Richer provenance, longer deployments, independent frameworks, and prospective threshold selection remain important future implementations.

\section{Full Gate Threshold Sweep}

Table~\ref{tab:appendix-gate} reports all evaluated thresholds. At $k=2$ and $k=3$, the full and conservative graphs agree and add no benign blocks. Larger thresholds increase attack coverage but also impose a growing benign cost.

\begin{table}[tbp]
  \centering
  \caption{Additional decisions blocked relative to the sequence-only gate.}
  \label{tab:appendix-gate}
  \small
  \begin{tabular}{lrrrr}
    \toprule
    & \multicolumn{2}{c}{Full graph} & \multicolumn{2}{c}{Conservative} \\
    \cmidrule(lr){2-3}\cmidrule(lr){4-5}
    $k$ & Attack & Benign & Attack & Benign \\
    \midrule
    2  & +5  & +0  & +5  & +0 \\
    3  & +5  & +0  & +5  & +0 \\
    5  & +23 & +3  & +7  & +1 \\
    8  & +48 & +15 & +21 & +5 \\
    10 & +71 & +34 & +36 & +14 \\
    \bottomrule
  \end{tabular}
\end{table}

\section{Graph Construction Procedure}
\label{app:construction}

This section spells out the measurement pipeline at the level needed to reimplement it. Each trajectory is first normalized into an ordered list of events. Every event receives a stable identifier and retains its event type, tool name, arguments, returned identifiers, state key, source identifiers, and position in the trace. No free-text similarity or model judgment is used.

For each trajectory, graph construction proceeds as follows:

\begin{enumerate}\itemsep2pt\parskip0pt
  \item Create one node for every normalized event and connect consecutive nodes with sequence edges.
  \item Join each tool call to its recorded result. This edge class is included only in the full graph.
  \item Index writes by state key and join each later read to the preceding write or writes recorded for that key.
  \item Index emitted typed identifiers and join a later event when its recorded \texttt{source\_ids} field contains an exact identifier match.
  \item Join tool calls emitted in the same assistant batch to their shared decision event.
  \item For every injection--sink pair, compute the shortest-path length in the sequence graph and in both influence-graph variants. Store $\DT$, $\DI$, and $\Gap=\DT-\DI$ together with trace-level outcome labels.
\end{enumerate}

The conservative graph contains the identical pipeline except that it does not perform step 2. All 4 kinds of graphs, having identical sequence edges, show that the sequence graph is spanned by the influence graphs, thus $\DI \le \DT$ is a construction invariant, and useful for us to implement. Whenever $\DI > \DT$, the graph is invalid or index is wrong.

\section{Measurement Record Summary}

Table~\ref{tab:schema} lists the fields of the pair-level analysis table, which holds one row per comparable injection-sink pair; all reported aggregates are reproducible from it.

\begin{table}[tbp]
  \centering
  \caption{Core fields in the pair-level analysis table.}
  \label{tab:schema}
  \small
  \begin{tabular}{>{\raggedright\arraybackslash}p{0.27\linewidth}p{0.65\linewidth}}
    \toprule
    Field & Meaning \\
    \midrule
    \texttt{run\_id} & Stable trajectory identifier. \\
    \texttt{suite}, \texttt{backend}, \texttt{model\_id} & Domain and model provenance. \\
    \texttt{task\_id}, \texttt{attack\_id} & User task and injection-task identifiers. \\
    \texttt{injection\_id}, \texttt{sink\_id} & Event identifiers defining the measured pair. \\
    \texttt{d\_t} & Shortest-path length in the sequence graph. \\
    \texttt{d\_i\_full}, \texttt{d\_i\_con} & Influence distances under full and conservative graphs. \\
    \texttt{gap\_full}, \texttt{gap\_con} & Corresponding values of $\DT-\DI$. \\
    \texttt{attack\_success}, \texttt{task\_success} & Corrected security and utility outcomes. \\
    \texttt{n\_events} & Length of the normalized trajectory. \\
    \texttt{influence\_edge\_types} & Edge classes appearing on the selected influence path. \\
    \bottomrule
  \end{tabular}
\end{table}

\section{Complete Decoupling Statistics}

Table~\ref{tab:complete-decoupling} expands the abbreviated main-text result. It reports both graph variants for all long-horizon backends, as well as the conservative banking graph used for cross-domain comparison.

\begin{table}[tbp]
  \centering
  \caption{Complete per-model decoupling summary. Medians refer to the pair-level distributions.}
  \label{tab:complete-decoupling}
  \scriptsize
  \begin{tabular}{llrrrrrr}
    \toprule
    Suite/model & Graph & Pairs & Dec. & Rate & Med. $\DT$ & Med. $\DI$ & Med. $\Gap$ \\
    \midrule
    LH GPT-4o-mini & Full & 127 & 116 & 91.3\% & 23 & 8 & 12 \\
    LH GPT-4o-mini & Conservative & 127 & 115 & 90.6\% & 23 & 9 & 9 \\
    LH GPT-4o & Full & 147 & 144 & 98.0\% & 24 & 14 & 7 \\
    LH GPT-4o & Conservative & 147 & 118 & 80.3\% & 24 & 17 & 2 \\
    LH Haiku 4.5 & Full & 91 & 91 & 100.0\% & 23 & 14 & 9 \\
    LH Haiku 4.5 & Conservative & 91 & 91 & 100.0\% & 23 & 18 & 4 \\
    LH Sonnet 4.6 & Full & 89 & 89 & 100.0\% & 23 & 14 & 8 \\
    LH Sonnet 4.6 & Conservative & 89 & 89 & 100.0\% & 23 & 18 & 2 \\
    LH all & Full & 454 & 440 & 96.9\% & 23 & 12 & 9 \\
    LH all & Conservative & 454 & 413 & 91.0\% & 23 & 16.5 & 3 \\
    \midrule
    Bank GPT-4o-mini & Conservative & 99 & 45 & 45.5\% & 6 & 6 & 0 \\
    Bank Haiku 4.5 & Conservative & 54 & 10 & 18.5\% & 2 & 2 & 0 \\
    Bank Sonnet 4.6 & Conservative & 78 & 23 & 29.5\% & 5 & 5 & 0 \\
    Bank all & Conservative & 231 & 78 & 33.8\% & 5 & 5 & 0 \\
    \bottomrule
  \end{tabular}
\end{table}

The model-level rows provide clarity to two different aspects. First, there is no attack success that is responsible for long-horizon decoupling; for instance, while Anthropic models have a near-zero attack success, their decoupling value is still very high. In addition, it is possible to see that the level of banking decoupling depends on infrastructure, which means that, for instance, the Haiku banking traces are the shortest ones with the lowest decoupling levels. On the contrary, GPT-4o-mini produces longer traces, offering more opportunities for provenance edges to shorten the route.

\section{Statistical and Reproducibility Details}

Wilson intervals are used for proportions because several rates lie near zero or one. The E2 logistic regression remains restricted to the 274 OpenAI pairs: adding backends with zero or nearly zero successes causes complete or near-complete separation under ordinary maximum likelihood. We therefore avoid presenting unstable pooled coefficients as if they were comparable estimates. The regression uses the full-graph gap and controls for sequence distance, a binary delayed-attack indicator, and backend.

The gate evaluation is a paired counterfactual replay. For a fixed threshold, the sequence and influence gates see the same recorded decision. Since $\mathcal{B}_T(k)\subseteq\mathcal{B}_I(k)$ analytically, discordance can occur only when the influence gate blocks an additional action. We report those additional actions separately for attack and benign decisions rather than claiming an unconstrained head-to-head win rate.

\paragraph{Code availability.}
The code required to reproduce the analyses and figures will be released publicly upon acceptance.

\end{document}